\newcommand{\aamasoff}[1]{}

\documentclass[submission,copyright,creativecommons]{eptcs}
\providecommand{\event}{SR' 2013} 
\usepackage{breakurl}             

\title{How to Be Both Rich and Happy:  \\ Combining Quantitative and Qualitative Strategic Reasoning about Multi-Player Games (Extended Abstract)}
\author{Nils Bulling
\institute{
Clausthal University of Technology, Germany}
\email{bulling@in.tu-clausthal.de}
\and
Valentin Goranko
\institute{Technical University of Denmark, Denmark
}
\email{vfgo@imm.dtu.dk}
}
\def\titlerunning{Logics for quantitative and qualitative objectives}
\def\authorrunning{Nils Bulling \& Valentin Goranko}

\usepackage{xspace}
\usepackage{amsmath}
\usepackage{amssymb}
\usepackage{xcolor}
\usepackage{graphicx}

\newtheorem{definition}{Definition}

\newtheorem{proposition}{Proposition}
\newtheorem{lemma}{Lemma}

\newtheorem{corollary}{Corollary}
\newtheorem{example}{Example}
\newenvironment{proof}[1][]{\paragraph{Proof{#1}.}}{\hfill\null$\square$}

\newcommand{\Agents}{\ensuremath{\mathsf{Ag}}\xspace}
\newcommand{\States}{\ensuremath{\mathit{S}}\xspace}
\newcommand{\Actions}{\ensuremath{\mathsf{Act}}\xspace}
\newcommand{\model}{\ensuremath{\mathfrak{M}}\xspace}

\newcommand{\ACF}{\ensuremath{\mathsf{ACF}}\xspace}
\newcommand{\AC}{\ensuremath{\mathsf{AC}}\xspace}
\newcommand{\APCF}{\ensuremath{\mathsf{APCF}}\xspace}
\newcommand{\APC}{\ensuremath{\mathsf{APC}}\xspace}
\newcommand{\ut}{\mathsf{payoff}}

\newcommand{\nnote}[1]{{\ \color{blue}\bf{}{NB: #1}}}
\newcommand{\vnote}[1]{{\ \color{red}\bf{}{VG: #1}}}

\newcommand{\coop}[2][]{\langle\!\langle{#2}\rangle\!\rangle_{_{\!\mathit{#1}}}}

\newcommand{\Sometm}[1][]{\mathbf{F}}
\newcommand{\Always}[1][]{\mathbf{G}}
\newcommand{\Until}[1][]{\mathbf{U}}
\newcommand{\Next}[1][]{\mathbf{X}}
\newcommand{\true}{\top}

\newcommand{\Logicname}[1]{\ensuremath{\mathsf{#1}}}

\newcommand{\NGL}{\Logicname{NGL}\xspace}
\newcommand{\NGLs}{\Logicname{NGL^*}\xspace}

\newcommand{\Plays}{\ensuremath{\mathsf{Plays}}\xspace}
\newcommand{\Hists}{\ensuremath{\mathsf{Hist}}\xspace}
\newcommand{\NGM}{GCGMP\xspace}
\newcommand{\NGMs}{GCGMPs\xspace}
\newcommand{\AcTCM}{TCM\xspace}
\newcommand{\automat}{\ensuremath{A}\xspace}
\newcommand{\automathalts}{\ensuremath{\automat\!\downarrow}\xspace}
\newcommand{\prop}[1]{#1}

\newcommand{\D}{\mathbb{D}}
\newcommand{\R}{\mathbb{R}}

\newcommand{\out}{\ensuremath{\mathsf{out}\xspace}}

\newcommand{\stat}{\ensuremath{\mathsf{St}\xspace}}
\renewcommand{\States}{\stat}
\newcommand{\dd}{\ensuremath{\mathsf{act}\xspace}}

\newcommand{\aga}{\ensuremath{\mathsf{a}\xspace}}
\newcommand{\agb}{\ensuremath{\mathsf{b}\xspace}}
\newcommand{\acta}{\ensuremath{\alpha\xspace}}
\newcommand{\one}{\ensuremath{\mathsf{1}\xspace}}
\newcommand{\agk}{\ensuremath{\mathsf{k}\xspace}}
\newcommand{\cgs}{\ensuremath{\mathcal{S}\xspace}}

\newcommand{\cut}[1]{}
\newcommand{\vcut}[1]{}

\newcommand{\keyterm}[1]{\emph{#1}}

\newcommand{\powerset}[1]{\mathcal{P}({#1})}

\newcommand{\lab}{\mathsf{L}}
\newcommand{\Props}{\ensuremath{\mathsf{Prop}}}   
\newcommand{\props}{\ensuremath{\mathsf{Prop}}}   

\renewcommand{\vec}[1]{\overrightarrow{#1}}
\newcommand{\cnf}{\ensuremath{\mathsf{Con}}}

\newcommand{\atlu}{\, \mathbf{U} \, }
\newcommand{\atlx}{\mathord \mathbf{X}}

\newcommand{\atlg}{\mathord \mathbf{G}}

\newcommand{\bigpaper}[1]{}
\newcommand{\ATL}{\Logicname{ATL}\xspace}
\newcommand{\ATLs}{\Logicname{ATL^*}\xspace}
\newcommand{\LTL}{\Logicname{LTL}\xspace}
\newcommand{\CTL}{\Logicname{CTL}\xspace}
\renewcommand{\NGL}{\Logicname{QATL}\xspace}
\renewcommand{\NGLs}{\Logicname{QATL^*}\xspace}
\newcommand{\up}{\vspace{0cm}}

\begin{document}
\maketitle

\vcut{LONG ABSTRACT

\begin{abstract}
Game theory and logic have developed two different traditions in studying
abilities of agents to achieve objectives in multi-player games. Game
theory has been studying abilities of rational players to achieve
\emph{quantitative objectives}: optimizing payoffs or, more generally, 
preferences on outcomes. On the other hand, logic has mainly been dealing with
abilities of players for achieving \emph{qualitative objectives} of players:
reaching or maintaining game states with desired properties, e.g., winning
states or safe states.  Put as a slogan, the former tradition is concerned with how a player can become maximally rich, or pay the least possible price, in
the game, while the latter tradition -- with how a player can achieve a
state of `happiness', or avoid reaching a state of `unhappiness', in the
game. This paper purports to combine the two agendas in a common logical
framework by enriching concurrent game models with payoffs for the normal
form games associated with the game states, thus enabling the combination
of quantitative and qualitative reasoning. Again  put as a slogan, our
framework enables one to reason about whether and how a player can reach or
maintain a state of `happiness' while becoming or remaining maximally
rich, or paying a minimal cost on the way.
\end{abstract}
}

\begin{abstract}
We propose a logical framework combining a game-theoretic study of abilities of agents to achieve \emph{quantitative objectives} in multi-player games by optimizing payoffs or preferences on outcomes with a logical analysis of the abilities of players for achieving \emph{qualitative objectives} of players, i.e., reaching or maintaining game states with desired properties. We enrich concurrent game models with payoffs for the normal form games associated with the states of the model and propose a quantitative extension of the logic \ATLs enabling the combination of quantitative and qualitative reasoning. 
\end{abstract}

\section{Introduction}

There are two rich traditions in studying strategic abilities of agents 
in multi-player games:

 \emph{Game theory}  has been studying rational behavior of players, relevant for their  achievement of \emph{quantitative objectives}: optimizing payoffs (e.g., maximizing rewards or minimizing cost) or, more generally, preferences on outcomes. Usually, the types of games studied in game theory are one-shot normal form games, their (finitely or infinitely) repeated versions, and extensive form games.

\emph{Logic}  has been mostly dealing with strategic abilities of players for achieving \emph{qualitative objectives}: reaching or maintaining outcome states with  desired properties, e.g., winning states, or safe states, etc.

Among the most studied models in the logic tradition are \emph{concurrent game models}~\cite{Alur02ATL,Pauly02modal}.  On the one hand they are richer than normal form games, as they incorporate a whole family of such games, each associated with a state of a transition system; but on the other hand, they are somewhat poorer because the outcomes of each of these normal form games, associated with a given state, are simply the successor states with their associated games, etc. whereas no payoffs, or even preferences on outcomes, are assigned. Thus, plays in concurrent game models involve a sequence of possibly different one-shot normal form games played in succession, and all that is taken into account in the purely logical framework are the properties -- expressed  by formulae of a logical language -- of the states occurring in the play.
Concurrent game models can also be viewed as generalization of (possibly infinite) extensive form games where cycles and simultaneous moves of different players are allowed, but no payoffs are assigned. 

Put as a slogan, the game theory tradition is concerned with \emph{how a player can become maximally rich, or how to pay as little cost as possible}, while the logic tradition -- with \emph{how a player can achieve a state of `happiness', e.g. winning, or to avoid reaching a state of `unhappiness' (losing) in the game}.

The most essential technical difference between qualitative and quantitative players' objectives is that the former typically refer to (a temporal pattern over) Boolean properties of game states on a given play and can be monitored locally 
whereas the latter are determined by the entire history of the play (accumulated payoffs) or even the whole play (its value, being a limit of average payoffs, or of discounted accumulated payoffs). It is therefore generally computationally more demanding and costly to design strategies satisfying quantitative objectives or to verify their satisfaction under a given strategy of a player or coalition.  


These two traditions have followed rather separate developments, with generally quite different agendas, methods and results, including, inter alia: 

\begin{itemize}
\itemsep = 0pt
\item on the purely qualitative side, \emph{logics of games and multiagent systems}, such as the Coalition logic CL \cite{Pauly02modal}, the Alternating time temporal logic ATL \cite{Alur02ATL}, and variations of it, see e.g. 
 \cite{Goranko04comparingKRA}, 
 \cite{Jamroga07constructive-jancl}, etc., formalizing and studying qualitative reasoning in  concurrent game models;  
\item  some \emph{single-agent and multi-agent bounded resource logics}  
\cite{BullingFarwer2010RAL-ECAI-short,AlechinaLNR11,MonicaNP11} 
 extending or modifying concurrent game models with some quantitative aspects by considering cost of agents' actions and reasoning about what players with bounded resources can achieve.  
 \item extensions of qualitative reasoning (e.g., reachability and B\"uchi objectives) in multi-player concurrent games with 'semi-quantitative' aspects by considering a preference preorder on the set of qualitative objectives, see e.g., \cite{BouyerBMU11}, \cite{BouyerBMU12}, thereby adding payoff-maximizing objectives and thus creating a setting where traditional game-theoretic issues such as game value problems and Nash equlibria become relevant.  
\item  deterministic or stochastic \emph{infinite games on graphs}, with qualitative objectives: typically, reachability, and more generally -- specified as $\omega$-regular languages over the set of plays, see e.g.  \cite{AlfaroHK07}, \cite{ChatterjeeAH11}, 
\cite{ChatterjeeH12}. 
\item on the purely quantitative side, first to mention \emph{repeated games}, extensively studied in game theory (see e.g., \cite{Osborne94gamet}), which can be naturally treated as simple, one-state concurrent game models with accumulating payoffs paid to each player after every round and no qualitative objectives;
\item from a more computational perspective, 
stochastic games with quantitative objectives on discounted, mean or total payoffs, in particular energy objectives, see e.g. \cite{ChatterjeeD12}.
\item the conceptually different but technically quite relevant study of \emph{counter automata, Petri nets, vector addition systems}, etc. -- essentially a study of the purely quantitative single-agent case of concurrent game models (see e.g. \cite{Esparza98decidabilityand}),  
\vcut{Blockelet2011model_checking_vas)} 
 where only accumulated payoffs but no qualitative objectives are taken into account and a typical problem is to decide reachability of payoff configurations satisfying formally specified arithmetic constraints from a given initial payoff configuration.   
\end{itemize}

A number of other relevant references discuss the interaction between qualitative and quantitative reasoning in multi-player games, e.g.  \cite{Pinchinat07}, \cite{GU08}, which we cannot discuss here due to space limitations.

This project purports to combine the two agendas in a common logical framework, by enriching concurrent game models with payoffs for the one-shot normal form games associated with the states, and thus enabling the combination of quantitative game-theoretic reasoning with the qualitative logical reasoning. 
Again, put as a slogan, our framework allows reasoning about whether/how a player can reach or maintain a state of `happiness' while becoming, or remaining, as rich as (rationally) possible, or paying the least possible price on the way. 
The purpose of this extended abstract is to introduce and discuss a general framework of models and logics for combined quantitative and qualitative reasoning that would naturally cover each of the topics listed above, and to initiate 
a long term study on it. 
\vcut{
 inter alia focusing on identifying semantic and syntactic conditions under which the model checking problem in this framework is decidable. 
 }

\section{Preliminaries} 
\label{sect:prelim}


A \textbf{concurrent game model}~\cite{Alur02ATL} (CGM) 
$\cgs = (\Agents, \stat, \{\Actions_\aga\}_{\aga\in\Agents}, \{\dd_{\aga}\}_{\aga\in\Agents}, \out,\Props,\lab)$
comprises: 
\vcut{
}
\begin{itemize}\up
\itemsep = -3pt
\item a non-empty, fixed 
set of players $\Agents=\{1,\dots, k\}$ and a set of actions $\Actions_{\aga}\neq \emptyset$ for each $\aga\in\Agents$. \\
For any $A\subseteq\Agents$ we will denote $\Actions_A:=\prod_{\aga\in A}\Actions_{\aga}$ and will use $\vec{\acta}_{A}$ to denote a tuple from $\Actions_A$. 
 In particular,  $\Actions_\Agents$ is the set of all possible action profiles in $\cgs$. \up

\item a non-empty set of  game states $\stat$.
%
\vcut{
}
\item for each $\aga\in\Agents$ a map $\dd_{\aga} : \stat \rightarrow \powerset{\Actions_{\aga}}$ setting for each state $s$
the actions available to $\aga$ at $s$. 
\item a transition function $\out: \States\times\Actions_\Agents\,  \rightarrow \States$ that assigns the (deterministic) \keyterm{successor (outcome) state} $\out(q,\vec{\acta}_{\Agents})$ to every state $q$ and action profile $\vec{\acta}_{\Agents} = \langle\alpha_{\one}, \dots, \alpha_{\agk}\rangle$ such that $\alpha_{\aga} \in \dd_{\aga}(q)$ for every $\aga \in \Agents$ (i.e., every $\alpha_{\aga} $ that can be executed by player $\aga$ in state $q$).\up 
\item  a  set of atomic propositions $\Props$ and a labelling function $\lab:\stat\rightarrow \powerset{\Props}$.
%
\end{itemize}


Thus, all players in a CGM execute their actions synchronously and the combination of these actions, together with the current state, determines the transition to a (unique) successor state in the CGM.\smallskip

The \textbf{logic of strategic abilities \ATLs} (\emph{Alternating-Time Temporal Logic}), introduced and studied 
in \cite{Alur02ATL}, is a logical system, suitable for specifying and verifying qualitative objectives of players and coalitions in concurrent game models.
The main syntactic construct of \ATLs is a formula of type $\coop{C} \gamma$, intuitively meaning: 
\textit{``The coalition $C$ has a collective strategy to guarantee the satisfaction of the objective $\gamma$ on every play enabled by that strategy.''}
Formally, \ATLs is a multi-agent extension of the branching time logic CTL*, i.e., multimodal logic extending the linear-time temporal logic \LTL -- comprising the temporal operators $\atlx$ (``at the next state''), $\atlg$ (``always from now on'') and $\atlu$ (``until'') -- with \emph{strategic path quantifiers} $\coop{C}$ indexed with coalitions $C$ of players. There are two types of formulae of \ATLs, \emph{state formulae}, which constitute the logic and  that are evaluated at game states,  and \emph{path formulae}, that are evaluated on game plays. These are defined by mutual recursion with the following grammars, where $C\subseteq \Agents$, $\prop{p}\in\Props$: state formulae are defined by $\varphi::= \prop{p} \mid \neg \varphi \mid \varphi\wedge\varphi \mid \coop{C}\gamma$,
and path formulae by  $\gamma::=\varphi \mid \neg\gamma \mid \gamma\land\gamma \mid \atlx \gamma \mid \atlg \gamma \mid \gamma \atlu \gamma$. 



The logic \ATLs is very expressive and that comes at a high computational price: satisfiability and model checking are $\mathbf{2ExpTime}$-complete. A computationally better behaved fragment is the logic \ATL,  which is the multi-agent analogue of CTL, only involving state formulae defined by the following grammar, for $C\subseteq \Agents$, $\prop{p}\in\Props$:
$\varphi::= \prop{p} \mid \neg \varphi \mid \varphi\land \varphi \mid \coop{C} \atlx \varphi \mid \coop{C} \atlg \varphi \mid \coop{C}(\varphi \atlu \varphi)$. For this logic satisfiability and model checking are $\mathbf{ExpTime}$-complete and $\mathbf{P}$-complete, respectively. 
We will, however, build our extended logical formalism on the richer \ATLs because we will essentially need the path-based semantics for it.   
%

%

\smallskip

\textbf{Arithmetic Constraints. }
We  define a simple language of arithmetic constraints to express conditions about the accumulated payoffs of players on a given play. For this purpose, we use a set $V_{\Agents}=\{v_{\aga}\mid \aga\in \Agents\}$ of special variables to refer to the accumulated payoffs of the players at a given state and  denote by $V_A$ the restriction of $V_{\Agents}$ to any group $A\subseteq \Agents$. 
The payoffs can be integers, rationals\footnote{Note that models with rational payoffs behave essentially like models with integer payoffs, after once-off initial re-scaling.}, or any reals. We denote the domain of possible values of the payoffs, assumed to be a subset of the reals $\R$, by $\D$ and use a set of constants symbols $X$, with $0\in X$, for names of special real values (see further) 
to which we want to refer in the logical language. 

%
%
For fixed sets $X$ and  $A\subseteq \Agents$ we build the set $T(X,A)$ of \emph{terms over $X$ and $A$} from $X\cup V_A$ by applying addition, e.g. $v_a+v_b$. An evaluation of a term $t\in T(X,A)$ is a mapping $\eta:X\cup V_A\rightarrow \D$. We write $\eta\models t$ to denote that $t$ is \emph{satisfied} under the evaluation $\eta$. Moreover, if some order of the elements $X\cup V_A$ is clear from context, we also represent an evaluation as a tuple from $\D^{|A|+|V_A|}$ and often assume that elements from $X$ have their canonic interpretation.
The set $\AC(X,A)$ of \emph{arithmetic constraints} over $X$ and $A$ consists of all expressions of the form $t_1 * t_2$ where $* \in \{<, \leq, =, \geq, >\}$ and $t_1,t_2\in T(X,A)$. We use $\ACF(X,A)$ to refer to the set of Boolean formulae over $\AC(X,A)$; e.g. $(t_1<t_2)\wedge (t_2\geq t_3)\in\ACF(X,A)$ for $t_1,t_2,t_3\in T(X,A)$. We note that the language  $\ACF(X,A)$ is strictly weaker than Presburger arithmetic, as it involves neither quantifiers nor congruence relations. 
%
%
%

We also consider the set $\APC(X,A)$ of 
 \emph{arithmetic path constraints} being expressions of the type $w_{a} * c$ where $a \in \Agents$, $* \in \{<, \leq, =, \geq, >\}$ and $c \in X$. The meaning of $w_{a}$ is to represent the value of the current play for the player $a$. That value can be defined differently, typically by computing the accumulated payoff over the entire play, by using a future discounting factor, or by taking the limit -- if it exists -- of the mean (average) accumulated payoff (cf.~\cite{Osborne94gamet}).
 We note that the discounted, accumulated, mean or limit payoffs may take real values beyond the original domain of payoffs $\D$; so, we consider the domain for $X$ to be a suitable closure of  $\D$.     
\vspace{-0.3cm}

\section{Concurrent Game Models with Payoffs and Guards}\vspace{-0.3cm}

We now extend concurrent game models with 
utility values for every action profile applied at every state and with guards that determine which actions are available to a player at a given configuration, consisting of a state and a utility vector,   in terms of arithmetic constraints on the utility of that player.   

\begin{definition} 
A  \textbf{guarded CGM with payoffs (\NGM)} 
is a tuple 
$\model=(\cgs,  
\ut, \{g_{\aga}\}_{\aga\in\Agents}, \{d_{\aga}\}_{\aga\in\Agents})$ 
where $\cgs = (\Agents, \stat, \{\Actions_\aga\}_{\aga\in\Agents}, \{\dd_{\aga}\}_{\aga\in\Agents}, \out,\Props,\lab)$ is a CGM and: 
\vcut{
}
\begin{itemize}\up
\vcut{
%
%
}
\itemsep = 0pt
\item $\ut:\Agents\times \States\times \Actions_\Agents\rightarrow  \D$ is a payoff function assigning at every state $s$ and action profile applied at $s$ a payoff to every agent. 
We  write $\ut_{\aga}(s,\vec{\acta})$ for $\ut(\aga,s,\vec{\acta})$. \up
\item $g_{\aga}:\States\times \Actions_{\aga}\rightarrow \ACF(X,\{a\})$, for each player $\aga\in\Agents$,  is a guard function that assigns for each state $s\in \stat$ and action 
$\acta\in \Actions_{\aga}$
an arithmetic constraint formula $g_{\aga}(s,\acta)$ that determines whether  $\acta$ is available to $\aga$ at the state $s$ given the current value of $\aga$'s accumulated payoff. 
The guard must enable at least one action for $\aga$ at $s$. 
Formally, for each state $s\in\States$, the formula $\bigvee_{\acta\in\Actions_{\aga}} g_{\aga}(s,\acta)$ must be valid. Moreover, a guard $g_{\aga}(s,\acta)$ is called \emph{state-based} if $g_{\aga}(s,\acta)\in\ACF(X)$.
\up
\item $d_{\aga}\in[0,1]$ is a discount factor, for each $\aga\in\Agents$, used in order to define realistically values of infinite plays for players or to reason about the asymptotic behavior of players' accumulated payoffs. 

\end{itemize}
\end{definition}
 
 The guard $g_{\aga}$ refines the function $\dd_{\aga}$ from the definition of a CGM, which can be regarded as a guard function assigning to every state and action a constant arithmetic constraint \texttt{true} or \texttt{false}. 
In our definition the guards assigned by $g_{\aga}$ only depend on the current state and the current accumulated payoff of  $\aga$. The idea is that when the payoffs are interpreted as costs, penalties or, more generally, consumption of resources the possible actions of a player would depend on her current availability of utility/resources. 

\begin{example}
\label{ex}
Consider the \NGM shown in Figure~\ref{fig:game} with 2 players, I and II, and 3 states, where in every state each player has 2 possible actions, $C$ (cooperate) and $D$ (defect). The transition function is depicted in the figure.
\begin{figure}
\begin{center}
\includegraphics[scale=0.45]{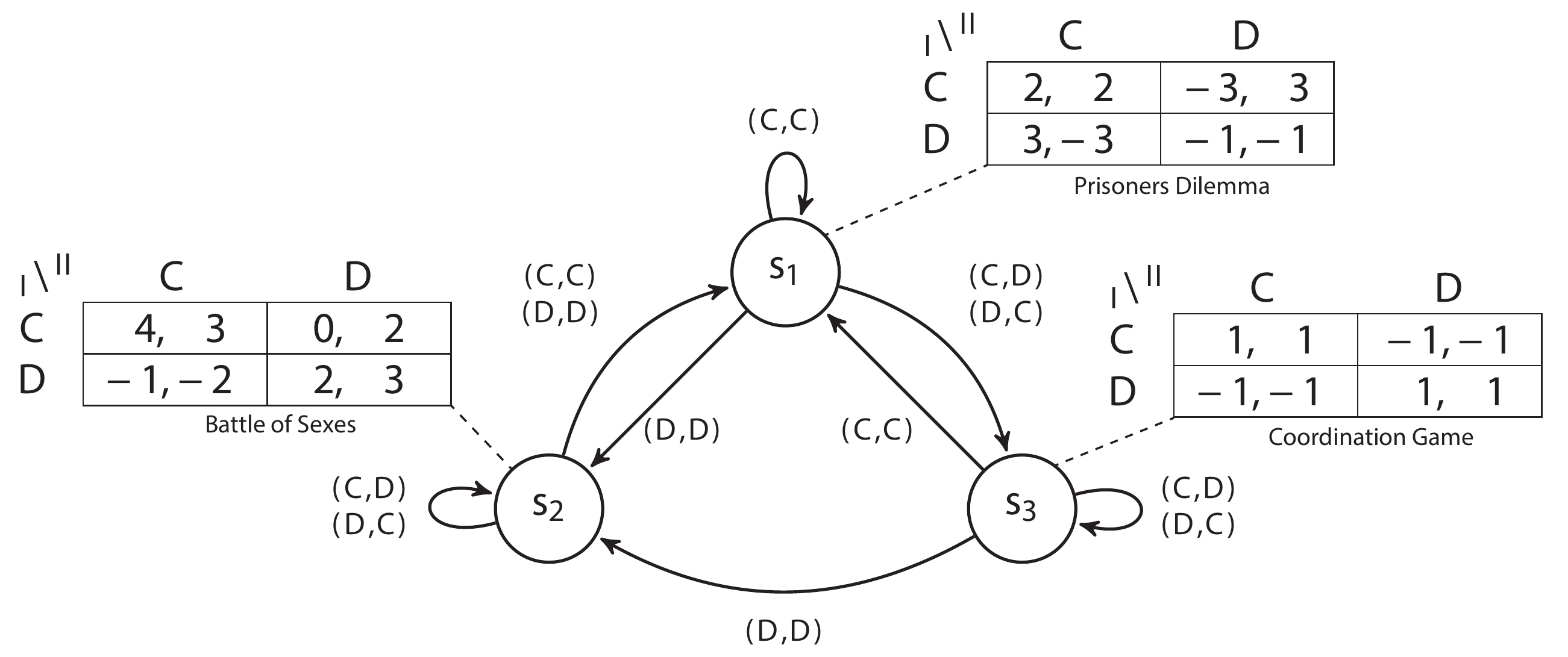}
\caption{A simple \NGM.}\label{fig:game}
\end{center}
\end{figure}
%
%
%
%
%
%
The normal form games associated with the states 
are respectively versions of the Prisoners Dilemma at state $s_{1}$, Battle of the Sexes at state $s_{2}$ and Coordination Game at state $s_{3}$. 

%
The guards for both players are defined at each state so that the player can apply any action if she has a positive current accumulated payoff, may only apply action $C$ if she has accumulated payoff 0; and  must  play an action maximizing her minimum payoff in the current game if she has a negative accumulated payoff. The discounting factors are 1 and the initial payoffs of both players are 0. 
\end{example} 


\textbf{Configurations, plays, and histories.} Let $\model$ be a \NGM defined as above. 
A \emph{configuration} (in $\model$) is a pair $(s,\vec{u})$ consisting of a state $s$ and a vector $\vec{u}=(u_1,\dots,u_k)$ of currently accumulated payoffs, one for each agent, at that state. Hereafter we refer to \emph{accumulated payoffs} as \emph{utility}, at a given state.
We define the set of possible configurations as $\cnf(\model)=\States\times\D^{|\Agents|}$. 
The \emph{partial configuration transition function} is defined as 
$\widehat{\out}: \cnf(\model) \times\Actions_\Agents\times\mathbb{N}\rightarrow \cnf(\model)$
 such that $\widehat{\out}((s,\vec{u}),\vec{\acta},l)=(s',\vec{u'})$ iff:  
 \begin{itemize}
 \itemsep = 0pt
\item[(i)]  $\out(s,\vec{\acta})=s'$ ($s'$ is a successor of $s$ if $\vec{\acta}$ is executed). 
\item[(ii)]  
\vcut{
$\widehat{u_a}\models g_{\aga}(s,\acta_{\aga})$  where $\widehat{u_a}$ is  of the that assigns $u_a$ to $v_a$ for each $\aga \in \Agents$
}
assigning the value $u_a$ to $v_a$ satisfies the guard $g_{\aga}(s,\acta_{\aga})$ for each $\aga \in \Agents$, i.e. $u_a\models g_{\aga}(s,\acta_{\aga})$ (each agent's move $\acta_{\aga}$ is enabled at $s$ by the respective guard $g_{\aga}$ applied to the current accumulated utility value $u_a$).
\item[(iii)] $u'_{\aga}=u_{\aga}+d_{\aga}^l\cdot \ut_{\aga}(s,\vec{\acta})$  for all $\aga\in\Agents$ (i.e., the utility values change according to the utility function and the discounting rate where $l$ denotes the number of steps that took place). 
\end{itemize}

A \NGM $\model$ with a designated initial configuration $(s_{0},\vec{u_0})$ 
gives rise to a 
 \emph{configuration graph on $\model$} consisting of all configurations in $\model$ reachable from $(s_{0},\vec{u_0})$ by the configuration transition function. 
%
%
\cut{
The next proposition shows that the players do always have a way to move.

\begin{proposition}For any \NGM the transition functions $\out$ and $\widehat{\out}$  are serial. \vnote{I think this can be omitted.} 
\end{proposition}
}
%
\vcut{
We note that for any \NGM the transition functions $\out$ and $\widehat{\out}$  are serial and 
define a play in an \NGM as an infinite sequence of interleaved configurations and actions.
}
%
A \emph{play} in a \NGM $\model$ is an infinite sequence $\pi=c_0 \vec{\acta_0},c_1 \vec{\acta_1},\dots$ from $(\cnf(\model)\times\Actions)^\omega$ such that $c_{n}\in\widehat{\out}(c_{n-1}, \vec{\acta}_{n-1})$ for  all $n>0$. The set of all plays in $\model$ is denoted by $\Plays_{\model}$. Given a play $\pi$ we use $\pi[i]$ and $\pi[i,\infty]$ to refer to the $i$th element and to the subplay starting in position $i$ of $\pi$, respectively.%
\bigpaper{Using the notation introduced in Sect.  
\ref{sect:prelim}, for each $i\geq 0$ we use $\pi[i]$ to refer the $i$th pair $(c_i,\vec{\acta_i})$ on $\pi$ and $\pi[i,\infty]=c_{i}\acta_{i},c_{i+1}\acta_{i+1},\dots$ to denote the subplay of $\pi$ starting from position $i$.
}
A \emph{history} is any finite initial sequence $h=c_0 \vec{\acta_0},c_{1}\acta_1, \dots, c_n\in (\cnf(\model)\times\Actions)^*\cnf(\model)$ of a play in $\Plays_{\model}$. 
\bigpaper{
Note that a history ends in a configuration. Sometimes we also assume that $\vec{\acta}_n$ is present and equals $\epsilon$ for technical reasons. 
}
The set of all histories is denoted by $\Hists_{\model}$. For any history $h$ we also define $h[i]$ as for plays and additionally $h[\mathit{last}]$ and $h[i,j]$ to refer to the last state on $h$ and to the sub-history between $i$ and $j$, respectively. 
\vcut{
Like for plays, we use the notation $h[i]$ and $h[i,j]$, $i\leq j$, $i<|h|$, to refer to the sub-history $h[i]\dots h[min(j,|h|-1)]$. We also allow $j=\infty$ and note that $h[\infty]$ refers to the last pair $(c_n,\epsilon)$. \\
\vnote{I suggest using $h[last]$ or $h[*]$  instead.}
}
Finally, we introduce functions $\cdot^c$,  $\cdot^u$, and $\cdot^s$ which denote the projection of a given play or history to the sequence of its  configurations, utility vectors, and states, respectively. For illustration, let us consider the play $\pi=c_0 \vec{\acta_0},c_1 \vec{\acta_1},\dots$. We have that $\pi[i,\infty]= c_i \vec{\acta_i},c_{i+1} \vec{\acta_{i+1}},\dots$; $\pi[i]=c_i\vec{\alpha}_i$; $\pi^c[i,\infty]= c_i,c_{i+1},\dots$;  $\pi^c[i]=c_i$; $\pi^a[i]=\vec{\acta_i}$; $\pi^u[i]=v_i$; and  $\pi^s[i]=s_i$  where $c_i=(s_i,\vec{u_i})$.

\begin{example}
\label{ex2}
Some possible plays starting from $s_{1}$ in Example \ref{ex} are given in the following where we assume that the initial accumulated payoff is $0$ for both agents. We note that this implies that the first action taken by any agent is always $C$.
\begin{enumerate}
\itemsep = 0pt
\item Both players cooperate forever: 
$(s_{1}, 0,0), (s_{1}, 2,2), (s_{1}, 4,4), \ldots$     
\item After the first round both players defect and the play moves to $s_{2}$, where player I chooses to defect whereas II cooperates. Then I must cooperate while II must defect but at the next round can choose any action, so a possible play is:  
$(s_{1}, 0,0), (s_{1}, 2,2), (s_{2}, 1,1), (s_{2}, 0,-1), (s_{2}, 0,1), (s_{2}, 0,3), 
(s_{2},0,5), \ldots$
\item After the first round player I defects while II cooperates and the play moves to $s_{3}$, where they can get stuck indefinitely, until -- if ever -- they happen to coordinate, so a possible play is: \\ 
$(s_{1}, 0,0), (s_{1}, 2,2), (s_{3}, 5,-2), (s_{3}, 4,-3), (s_{3}, 3,-4), \ldots (s_{3}, 0,-7), (s_{3}, -1,-8), \ldots$. 

Note, however, that once player I reaches accumulated payoff 0 he may only apply $C$ at that round, so if player II has enough memory or can observe the accumulated payoffs of I he can use the opportunity to coordinate with I at that round  by cooperating, thus escaping the trap at $s_{3}$ and making a sure transition to $s_{2}$. 
\item If, however, the guards did not force the players to play C when reaching accumulated payoffs 0, then both players could plunge into an endless  misery if the play reaches $s_{3}$. 
\end{enumerate}
\end{example}
   


\textbf{Strategies. }A \emph{strategy} of a player $\aga$ is a function $s_{\aga} : \Hists\rightarrow \Actions$ such that if $s_{\aga}(h)=\alpha$ then $h^u[last]_\aga\models g_{\aga}(h^s[last],\alpha)$; that is, actions prescribed by a strategy must be enabled by the guard. 
%
Our definition of strategy is based on histories of configurations and actions, so it extends the notion of strategy from ~\cite{Alur02ATL} where it is defined on histories of states, and includes strategies, typically considered e.g. in the study of repeated games, where often strategies prescribe to the player an action dependent on the previous action, or history of actions, of the other player(s). Such are, for instance, \textsc{Tit-for-tat} or \textsc{Grim-trigger} in repeated Prisoners Dillemma; likewise for various card games, etc.  
Since our notion of strategy is very general, it easily leads to undecidable model checking problems. So, we also consider some natural restrictions, such as: \emph{state-based}, \emph{action-based} or \emph{configuration-based}, \emph{memoryless}, \emph{bounded memory}, of \emph{perfect recall} strategies 
\footnote{We note that all strategies need to be consistent with the guards, so state-based strategies are only applicable in models where the guards only take into account the current state, but not the accumulated payoffs.}.
Here we adopt a generic approach and assume that two classes of strategies $\mathcal{S}^p$ and $\mathcal{S}^o$  are fixed as parameters, with respect to which the proponents and opponents select their strategies, respectively. The proponent coalition $A$ selects a $\mathcal{S}^p$-strategy $s_A$ (i.e. one agreeing with the class $\mathcal{S}^p$) while the opponent coalition $\Agents\backslash A$ selects a $\mathcal{S}^o$-strategy $s_{\Agents\backslash A}$. The outcome play 
$\mathsf{outcome\_play}_{\model}(c,(s_A,s_{\Agents\backslash A}),l)$ 
 in a given \NGM $\model$ determines the play emerging from the execution of the (complete) strategy profile $(s_A,s_{\Agents\backslash A})$ from configuration $c$ in $\model$. 

\vspace{-0.3cm}

\section{The Logic: Quantitative ATL*}\vspace{-0.3cm} 

We now extend the logic \ATLs to the logic \NGLs with atomic quantitative objectives being state or path arithmetic constraints 
over the players' accumulated payoffs. The semantics of \NGLs naturally extends the semantics of \ATLs over \NGMs, but parameterised with the two classes of strategies  $\mathcal{S}^p$ and $\mathcal{S}^o$.  

\begin{definition}[The logic \NGLs]
The language of \NGLs consists of \emph{state formulae} $\varphi$, which constitute the logic, and \emph{path formulae}  $\gamma$, 
generated as follows, where $A\subseteq\Agents$, $\mathsf{ac}\in \AC$, $\mathsf{apc}\in\APC$, and $p\in \props$: \vspace{-0.5cm} 
\[\varphi::= {p} \mid \mathsf{ac}\mid \neg \varphi \mid \varphi\wedge\varphi
  \mid \coop{A}\gamma \ \mbox{and} \  
\gamma::=\varphi \mid \mathsf{apc}\mid\neg\gamma \mid \gamma\land\gamma \mid
  \Next\gamma \mid \Always\gamma\mid\gamma\Until\gamma.\]\vspace{-0.6cm}
%
%
\vcut{
The ``sometime'' operator can be defined as $\Sometm\gamma \equiv\true\Until\gamma$. 
}
%

Let $\model$ be a \NGM, $c$ a configuration,  $\varphi, \varphi_1, \varphi_2$ state-formulae, $\gamma,\gamma_1, \gamma_2$ path formulae, and 
$l\in\mathbb{N}$.  Further, let $\mathcal{S}^p$ and $\mathcal{S}^o$ be two classes of strategies as described above.
The semantics of the path constraints is specified according to the limit-averaging or discounting mechanism adopted for computing the value of a play for a player. Then the truth of a $\NGLs$ formula at a position of a configuration in $\model$ is defined by mutual recursion on state and path formulae as follows: 
\vspace{-0.05cm}
\begin{description}
\itemsep = 1pt
\item[$\model,c,l\models{p}$] for $p\in \props$ iff ${p}\in \lab(c^s)$;  \ $\model,c,l\models{\mathsf{ac}}$ for $\mathsf{ac}\in\AC$ iff $c^u\models \mathsf{ac}$,
\vcut{
\item[$\model,c,l\models{\alpha}$] iff $c^u\models \alpha$ and $\alpha\in\AC$,\up 
\item[$\model,c,l\models\varphi_1\wedge \varphi_2$] iff $\model,c,k\models\varphi_1$ and \ $\model,c,l\models\varphi_2$, \up
\item[$\model,c,l\models\neg \varphi$] iff not $\model,c,l\models \varphi$, \up
}
\item[$\model,c,l\models \coop{A}\gamma$] iff 
there is a collective $\mathcal{S}^p$-strategy $s_A$ for $A$ such that for all  collective $\mathcal{S}^o$-strategies $s_{\Agents\backslash A}$ for $\Agents\backslash A$ we have that $\model,{\mathsf{outcome\_play}}^{\model}(c,(s_A,s_{\Agents\backslash A}),l),l\models\gamma$.\up

%
%
\item[$\model,\pi,l\models \varphi$] iff  $\model,\pi[0],l\models \varphi$; \ $\model,\pi,l\models \mathsf{apc}$ iff  $\pi^u,l\models \mathsf{apc}$ for $\mathsf{apc}\in\APC$. 
\vcut{
\item[$\model,\pi,l\models \beta$] iff  $\pi^u,l\models \beta$ for $\beta\in\APC$. 
where the semantics of path constraints is specified separately, according to the limit-averaging or discounting mechanism adopted for computing the value of a play for a player.\up}
%
\item[$\model,\pi,l\models \Always\gamma$] iff $\model,\pi[i],l+i\models \gamma$ for all $i\in\mathbb{N}_0$,  \up
\item[$\model,\pi,l\models \Next\gamma$] iff $\model,\pi[1],l+1\models \gamma$, \up
\item[$\model,\pi,l\models  \gamma_1\Until \gamma_2$] iff there is $j\in\mathbb{N}_0$ such that $\model,\pi[j],l+j\models \gamma_2$ and $\model,\pi[i],l+i\models \gamma_1$ for all $0\leq i< j$.
%

%
\end{description}
Ultimately, we define $\model,c\models\varphi$ as $\model,c,1\models\varphi$. Moreover, if not clear from context, we also write $\models_{(\mathcal{S}^p,\mathcal{S}^o)}$ for $\models$.
\end{definition}

The semantics presented above extends the standard semantics for \ATLs and  is amenable to various refinements and restrictions, to be studied further. For instance, if appropriate, an alternative semantics can be adopted, based on irrevocable strategies \cite{Agotnes07irrevocable} or, more generally, on strategy contexts \cite{Brihaye08strategycontexts}  or other mechanisms for strategy commitment and release \cite{Jamroga08commitment-tr}. Also, the nested operators as defined here access the accumulated utility values and require plays to be infinite. Similarly to \cite{BullingFarwer2010RAL-ECAI-short}, \vcut{Blockelet2011model_checking_vas} one can consider variants of these settings  which may yield  decidable model checking and better complexity results.

%
%

%
As the logic \NGLs extends \ATLs, it allows expressing all purely qualitative \ATLs properties.
 It can also express purely quantitative properties, e.g.:  $\coop{\{\aga\}}\Always (v_{\aga} > 0)$ 
meaning ``Player $\aga$ has a strategy to maintain his accumulated payoff to be always positive'', or $\coop{A}(w_a \geq 3)$ meaning ``The coalition $A$ has a strategy that guarantees the value of the play for player $\aga$ to be at least 3''. Moreover, \NGLs can naturally express combined qualitative and quantitative properties, e.g. 
$\coop{\{\aga,\agb\}}((v_{\aga} + v_{\agb} \geq 1) \Until p)$), etc. 

\begin{example}
\label{ex3}
The following \NGLs state formulae are true at state $s_{1}$ of the \NGM in Example \ref{ex}, where $p_{i}$ is an atomic proposition true only at state $s_{i}$, for each $i=1,2,3$:  \\
(i) \ \ $\coop{\{I,II\}}  \Sometm(p_{1} \land v_{I} > 100 \land v_{II} > 100) 
\land   
\coop{\{I,II\}} \Next \Next \coop{\{II\}} (\Always(p_{2} \land v_{I} = 0) \ \land \ \Sometm\ v_{II} > 100)$. \\ 
(ii) \ \  $\lnot \coop{\{I\}} \Always (p_{1} \lor v_{I} > 0) \land \lnot \coop{\{I,II\}} \Sometm (p_{3} \land \Always (p_{3} \land (v_{I }+ v_{II} > 0))) $.
\end{example}

\vcut{
\begin{enumerate}
\itemsep = 0pt
\item  $\coop{I,II}  \Sometm(v_{I} > 100 \land v_{II} > 100)$  
\item  $\coop{I,II} \Next \Next \coop{II} (\Always\ v_{I} = 0 \ \land \ \Sometm\ v_{II} > 100)$ 
\item $\lnot \coop{I} \Always v_{I} > 0$
\item  $\coop{I,II} \Next \Next \lnot \coop{I} (\Sometm\ v_{I} \geq 0)$ 
\end{enumerate}
}

\vcut{
\section{Summary of preliminary results}  

Generally, the \NGM models are too rich and the language of \NGLs is too expressive to expect computational efficiency, or even decidability, of either model checking or satisfiability. Some preliminary results and related work show that model checking of \NGLs in \NGMs is undecidable under rather weak assumptions, e.g. if the proponents or the opponents can use memory-based strategies.  On the other hand, there are some natural semantic and syntactic restrictions of  \NGLs where decidability can be restored. These include restrictions similar to ones in~\cite{BullingFarwer2010RAL-ECAI-short}, as well as enabling only memoryless strategies, imposing non-negative payoffs, or constraints on the transition graph of the model, etc.  For lack of space we only sketch some results in an appendix; more will appear in a full paper under preparation. 
}

\section{(Un)Decidability: Related Work and Some Preliminary Results}  Generally, the \NGM models are too rich and the language of \NGLs is too expressive to expect computational efficiency, or even decidability, of either model checking or satisfiability testing. Some preliminary results and related work show that model checking of \NGLs in \NGMs is undecidable under rather weak assumptions, e.g. if the proponents or the opponents can use memory-based strategies. These undecidability results are not surprising as \NGMs are closely related to Petri nets and vector addition systems and it is known that  model checking over them is generally undecidable. In~\cite{Esparza_decidabilityof}, for example, this is shown for fragments of \CTL and (state-based) \LTL over Petri nets. Essentially, the reason is that the logics allow to encode a ``test for zero''; for Petri nets this means to check whether a place contains a token or not. In our setting  undecidability follows for the same reason, and we will sketch some results below. 
%
\vspace{-0.3cm}

\paragraph{Undecidability results.}
The logic $\NGL$ restricts $\NGLs$ in the same way as \ATL restricts \ATLs, due to lack of space we skip the formal definition. As a first result we show that model checking $\NGL$ is undecidable even if only the proponents are permitted to use perfect recall strategies and the opponents are bound to memoryless strategies. More formally, let $\mathsf{S}^{pr}$ denote the class of perfect recall state-based strategies and $\mathsf{S}^{m}$ the class of memoryless state-based strategies. That is, strategies of the former class are functions of type $\States^*\rightarrow\Actions$ and of the latter class  functions of type $\States\rightarrow \Actions$. 

Undecidability can be shown using ideas from e.g.~\cite{BullingFarwer2010RAL-ECAI-short,Esparza_decidabilityof}.
Here, we make use of the construction of~\cite{BullingFarwer2010RAL-ECAI-short} to illustrate the undecidability by simulating a two-counter machine (\AcTCM). A \AcTCM~\cite{HopcroftUllmann79Automata} can be considered as a transition system equipped with two integer counters that enable/disable transitions. Each  step of the machine depends on the current state, symbol on the tape, and the  counters, whether they are zero or not. After each step the counters can be incremented ($+1$), or decremented ($-1$) , the latter only if the respective counter is not zero. A \AcTCM  is essentially a (nondeterministic) push-down automaton with two stacks and exactly two stack symbols (one of them is the initial stack symbol) and has the same computation power as a Turing machine (cf.\ \cite{HopcroftUllmann79Automata}). 
A \emph{configuration} is a triple $(s,w_1,w_2)$ describing the current state ($s$), the value of counter 1 ($w_1$) and of counter 2 ($w_2$). A \emph{computation} $\delta$ is a sequence of subsequent configurations effected by transitions. 

For the simulation, we associate each counter with a player. The player's accumulated payoff encodes the counter value; actions model the increment/decrement of the counters; guards ensure that the actions respect the state of the counters. The accepting states of the two-counter machine are encoded by a special proposition $\prop{halt}$. Now, the following lemma stating the soundness of the simulation can be proved:

\begin{lemma}[Reduction]\label{lemma:reduction}
For any two-counter machine $A$ we can construct a finite  \NGM $\model^A$ with two players and  proposition $\mathsf{halt}$ such that the following holds: $A$ halts on the empty input iff $\model^A$ contains a play $\pi$ with $\pi^c=(s^0,(v_1^0,v_2^0))(s^1,(v_1^1,v_2^1))\dots$ such that there exists $j\in\mathbb{N}$ with $\mathsf{halt}\in \lab(s^j)$.
\end{lemma}

The next theorem gives two cases for which the model checking problem is undecidable. By the previous Lemma we have to ensure that the halting state is reached which can be expressed by $\coop{1}\Sometm\mathsf{halt}$. We can also use  purely state-based guards and encode the consistency checks in the formula as follows: $\coop{1}(v_1\geq 0\wedge v_2\geq 0\wedge e_1\rightarrow v_a=0\wedge e_2\rightarrow v_2=0)\Until\mathsf{halt}$ where the proposition $\prop{e_i}$ is added to the model to indicate that the value of counter $i$ is zero. Not that this information is static and obtained from the transition relation of the automaton. 

\begin{proposition}\label{cor:und1}
Model checking the logic $\NGL$ is undecidable, even for the 2 agent case and no nested cooperation modalities, where $\mathcal{S}^{p}=\mathsf{S}^{pr}$ and $\mathcal{S}^{o}=\mathsf{S}^{m}$. This does even hold either for formulae not involving arithmetic constraints, or for state-based guards.
\end{proposition}
\vspace{-0.5cm}

\paragraph{Restoring decidability.}
There are some natural semantic and syntactic restrictions of  \NGLs where decidability may be restored; these include for instance, the enabling of only memoryless strategies, imposing non-negative payoffs,  constraints on the transition graph of the model, bounds on players utilities etc. For instance, the main reason for the undecidability result above is the possibility for negative payoffs that allow for decrementing the accumulated payoffs and thus simulating the TCM operations. Therefore, a natural restriction in the quest for restoring decidability is to consider only  \NGM models with non-negative payoffs. In this case the accumulated payoffs increase monotonically over every play of the game, and therefore the truth values of every arithmetic constraint occurring in the guards and in the formula eventually stabilize in a computable way, which in the long run reduces the model checking of any $\NGL$-formula in an  \NGM to a model checking of an \ATL-formula in a CGM. One can thus obtain decidability of the model checking of the logic $\NGL$ in finite \NGM with non-negative payoffs and perfect information. We will discuss these and other decidability results in a future work, where we will also consider restrictions similar to~\cite{BullingFarwer2010RAL-ECAI-short}.

\vcut{
We therefore obtain the following: 

\begin{proposition}\label{prop1}
Model checking the logic $\NGL$ is decidable in any finite \NGM with non-negative payoffs. 
\end{proposition}
}

\vspace{-0.3cm}

\section{Concluding Remarks}\vspace{-0.3cm} This paper proposes a long-term research agenda bringing together issues, techniques and results from several research fields. It aims at bridging the two important aspects of reasoning about objectives and abilities of players in multi-player games: quantitative and qualitative, and eventually providing a uniform framework for strategic reasoning in multi-agent systems. 
\smallskip

\textbf{Acknowledgements:} We thank the anonymous referees for detailed and helpful comments and additional references. 
\vspace{-0.3cm}

\bibliographystyle{eptcs}
\bibliography{generic.bib}

\appendix

\vcut{\appendix

\nnote{Todo}We firstly consider two undecidability results. In~\cite{Esparza94onthe} it is shown that model checking the weak branching time logic $UB^-$ (i.e. CTL with operator $EF$) over Petri nets is undecidable. This result does even hold in  case of no atomic propositions beside $\top$ (true). However, over restricted classes of Petri nets fragments turn out to be decidable (e.g. over nets whose set of reachable markings is effectively semilinear)~\cite{Esparza_decidabilityof}. Basic Parallel Processes (BPPs) and conflict-free Petri nets are models which fall into this class of restricted models. For the logic $UB$ (i.e. $UB^-$ enriched with operator $EG$), however, the model checking problem is undecidable already over BPPs~\cite{Esparza94decidabilityissues}. The authors of the latter paper conclude that \begin{quote}``no natural and useful branching time temporal logic for Petri nets seems to be decidable.''~\cite{Esparza94decidabilityissues}\end{quote}

The linear-time temporal logic LTL without atomic predicates (defined over configurations!) is shown to be decidable over Petri nets~\cite{Esparza94onthe}. However, if one adds propositions of type $s=0$ with the meaning that there are no tokens on place $s$, the logic becomes undecidable. This does even hold for very simple Petri nets corresponding to BPPs. Essentially, every extension of LTL  allowing to encode a test for zero is undecidable.

We consider the predicates $s\geq c$, $first(t)$, and $en(t)$ which are true at a marking $M$ of the Petri net if $M(s)\geq c$, the next transition after $M$ in the sequence is $t$, and if a transition of the Petri net $t$ is enabled by $M$. In~\cite{Esparza94onthe} it is shown that LTL with only $first(t)$ is decidable over Petri nets but undecidable over BPPs if additionally $s\geq c$ is added. 

Also the very basic fragment of LTL with only the temporal operator $F$ is undecidable  if all three predicates are present, already over conflict free Petri nets~\cite{Howell1989305}. In the following we refer to this logic as $L_F$. Of particular interest for our study are decidable fragments of $L_F$. The model checking problem was shown to be decidable if negation is only allowed in front of predicates~\cite{Howell1991341}. Decidability is also obtained if only the operator $GF$ is available and negation is only allowed in front of predicates~\cite{Jancar1990DTL}.
It is well known that Petri nets, VAS and VASS have the same expressivity power~\cite{hopcrof79VAS}. Hence, the (un)decidability results hold for all three models. The authors of~\cite{Blockelet2011model_checking_vas} consider an extension of CTL (without atomic propositions) with predicates $\mu(j)\geq c$ and temporal operators $U_\psi$ where $\psi$ is a Presburger arithmetic consttraint formula. The meaning of $\mu(j)\geq c$ is that $v_j\geq c$ where $(v_1,\dots,,v_n)$ are the values of the current configuration and $U_\psi$ constraints the possible plays. Essentially, the evaluation of the latter operator is based on the accumulated value of all transition vectors along a play fragment up to the current position. Clearly, given the results above, the model checking problem is undecidable in the general case, already if only the temporal operator $F$ is allowed. A positive result is that the \emph{existential fragment} without negation is  decidable. This result is also interesting for the fragments we are considering in this paper.

TODO: Refer to~\cite{BullingFarwer2010RAL-ECAI-short}. These restrictions seems most interesting for our purposes.

We note that the \emph{basic predicates} considered in the logics above take into account configurations. In our setting, however, atomic propositions are given wrt. the control state only!

\paragraph{Two-counter machines.} 
The following is based on~\cite{Bulling10phd,BullingFarwer2010RAL-ECAI-short} and is provided as a background.  
The undecidability proof is done by simulating  a \emph{two counter machine} $\automat$ (\AcTCM)  on the empty input. We write $\automathalts$ for  ``$\automat$ halts on empty input''.

Formally, a \AcTCM \automat is given by \[A=(S,\Gamma,s^\text{init},S_f,\Delta)\] where $S$ is a finite set of states, $\Gamma$ is the finite input alphabet, $s^\text{init}\in S$ is the initial state, $S_f\subseteq S$ is the set of final states, and 
\vcut{$\Delta\subseteq (S\times \Gamma\times \{0,1\}^2)\times(S\times \{-1,1\}^2)$} 
$\Delta\subseteq (S\times \Gamma\times \{0,1\}^2)\times(S\times \{-1,0,1\}^2)$ 
is the transition relation such that if $((s,a,E_1,E_2),(s',C_1,C_2))\in \Delta$ and $E_i=0$ then $C_i\neq -1$ for $i=1,2$ (to ensure that an empty counter cannot further be decremented). 
In the case of an empty input, we ignore the alphabet and assume 
\vcut{$\Delta\subseteq (S\times \{0,1\}^2)\times(S\times \{-1,1\}^2)$.}
$\Delta\subseteq (S\times \{0,1\}^2)\times(S\times \{-1,0,1\}^2)$.

A \AcTCM can be considered as a transition system equipped with two counters that influence the transitions. Each  step of the machine depends on the value of the counters, i.e. whether they are zero or non-zero. After each step the counters can be incremented or decremented.
It is important to stress that a \AcTCM can only distinguish between  an empty and non-empty counter. We consider the transition $((s,E_1,E_2),(s',C_1,C_2))\in\Delta$: $E_i=1$ (resp.\ $=0$) represents that counter $i$ is non-empty (resp.\ empty); $C_k=1$ (resp.\ $=-1$) denotes that counter $i$ is incremented (resp.\ decremented) by $1$, and $C_k=0$ that the counter value is not changed. The transition encodes that in state $s$ the TCM can change its state to $s'$ provided that the first (resp.\ second) counter meets  condition $E_1$ (resp.\ $E_2$). The value of counter $k$ changes according to $C_k$ for  $k=1,2$. The transition $((s,1,0),(s',-1,1))\in\Delta$, for example, is enabled if the current state is $s$, counter $1$ is non-empty, and counter $2$ is empty. If the transition is selected the state changes to $s'$, counter $1$ is decremented and counter $2$ is incremented by 1.

The general mode of operation is as for  pushdown automata. In particular, 

 A \emph{computation} $\delta$ is a sequence of subsequent configurations that can emerge by transitions according to $\Delta$ such that the first state is $s^\text{init}$.  An \emph{accepting} configuration is a finite computation $\delta=(s_i,w_1^i,w_2^i)_{i=1,\dots,k}$ where the last state $s_k$ is a final state, i.e. $s_k\in S_f$. We use $\delta_i=((s_i,E_1^{i},E_2^{i}),(s_{i+1},C_1^{i},C_2^{i}))\in \Delta$ to denote the tuple that leads from the $i$th configuration $(s_i,w_1^i,w_2^i)$ to the $i+1$st configuration $(s_{i+1},w_1^{i+1},w_2^{i+1})$ for $i<k$. In particular, we have that $w_j^{i+1}=w_j^i+C_j^i$ for $j=1,2$. 

\paragraph{General Idea of the Reduction.} We describe the general idea of the simulation of a $\AcTCM$ by means of a $\NGM$ together with a $\NGL$-formula. The formula is needed to ensure the correctness of the simulation. The model of a $\NGM$ on its own is too weak. Let $\automat=(S,\Gamma,s^\text{init},S_f,\Delta)$ be a \AcTCM. We represent the value of the two counters as the accumulated value  $w_1$ and $w_2$ of the two players $1$ and $2$, respectively. For each state of the machine, we add a state to the model and we label the accepting states in $S_f$ by proposition \prop{halt}. The increment and decrement  of counter values are modelled by transitions incrementing and decrementing from the corresponding utility value. The general idea underlying the reduction is as follows (the path formula depends on the specific fragment $\Logicname{L}$ of \NGLs considered): 

\begin{center}\emph{($\star$)\quad $\automathalts$ iff there is a play in the \NGM along  which a path formula $\gamma_\Logicname{L}$  is true.}\end{center} 

{\begin{figure}[t]
\centering
\includegraphics[scale=0.3]{figs/reduction_model.pdf}
\caption{Transformation of transitions $(s,E_1,E_2)\Delta (s_i,C_1,C_2)$ and $(s,E_1,E_2)\Delta (s_j,C_1',C_2')$ .}\label{fig:RALr}\label{fig:reduction}
\end{figure}}
The satisfying play in the \NGM  corresponds to an accepting computation of the TCM. 

It remains to construct formulae to ensure the sound modelling of the TCM transitions. For example, consider the tuple $((s,1,0),(s',-1,1))\in\Delta$. It can only be chosen if the second counter is  empty. This can directly be encoded into the model if \emph{guards} are available. If not, we need to delegate the task of the guards to the logic. The encoding of a transition $r:=((s,E_1,E_2),(s',C_1,C_2))$ is a three-step process (cf.\ Figure~\ref{fig:RALr}). In a state $s$ of the \NGM (we are economic and use the same notation for elements of the model and the TCM) an agent performs an action $\langle E_1, E_2\rangle$ in order to ``select'' the transition $r$; that is, to filter the transition of the TCM (note, that this action may not uniquely determine a transition of the TCM). The action results  in  a ``test'' state $s^{E_1E_2}$. In this state, an action $\langle s',C_1,C_2\rangle$  affecting the values of the counters $C_i$ / accumulated values $v_i$ can be executed (i.e.\ the action increments/decrements $v_i$). Clearly, such an action must only be enabled if the counter values are consistent with $E_i$ and $C_i$ respectively. The  check whether a counter/accumulated utility value is empty or not, takes place at the intermediate state $s^{E_1E_2}$. Therefore, we label the state with propositions $e_i$ iff $E_i=0$. Then, a logical formula $e_i\rightarrow (v_i=0)$ can be used to test wether a counter is empty or not.

\begin{lemma}[Reduction]\label{lemma:reduction}
For any two-counter machine $A$ we can construct a finite  \NGM $\model^A$ with two players and three propositions $\{\prop{e_1},\prop{e_2},\prop{halt}\}$ such that the following holds: $A$ halts on the empty input iff $\model^A$ contains a play $\pi$ with $\pi^c=(s^0,(v_1^0,v_2^0))(s^1,(v_1^1,v_2^1))\dots$ such that there exists $j\in\mathbb{N}$ with $halt\in \lab(s^j)$ and for all $i$ with $0\leq i<j$, $v_k^i\geq 0$ and whenever $e_k\in \lab(s^i)$ then $v_k^i=0$ for $k=1,2$.
\end{lemma}

\begin{corollary}\label{cor:und1}
Model checking the logic $\NGL$ is undecidable, even in the 2 agent case and formulae not involving arithmetic constraints nor nested cooperation modalities, where $\mathcal{S}^{p}=\mathsf{S}^{pr}$ and $\mathcal{S}^{o}=\mathsf{S}^{m}$.
\end{corollary}
\begin{proof}[{\ } (Sketch)] We modify the construction of $\model^\automat$ from Lemma~\ref{lemma:reduction} by refining the guards. We assume that an action $(s,E_1,E_2)$ is only applicable in state $s$ iff ($v_k=0$ iff $E_k=0$) for $k=1,2$; that is $g_1((s,E_1,E_2))=v_1\geq 0\wedge v_2\geq 0\wedge e_1\rightarrow v_a=0\wedge e_2\rightarrow v_2=0$. The resulting model is well-defined as action $\alpha_{err}$ is executable in all states. We have that $\coop{1}\Sometm\prop{halt}$ is true in $\model^A$ iff $\model^A$ contains a play $\pi$ with $\pi^c=(s^0,(v_1^0,v_2^0))(s^1,(v_1^1,v_2^1))\dots$ such that there is a $j\in\mathbb{N}$ such that $halt\in \lab(s^j)$ and for all $i$ with $0\leq i<j$, $v_k^i\geq 0$ and whenever $e_k\in \lab(s^j)$ then $v_k=0$ for $k=1,2$.
\end{proof}

\paragraph{Restoring decidability}

The main reason for the undecidability results above is the possibility for negative payoffs that allow for decrementing the accumulated payoffs and thus simulating the TCM operations. Therefore, a natural restriction in the quest for restoring decidability is to consider only  \NGM models with non-negative payoffs. In this case the accumulated payoffs increase monotonically over every play of the game, and therefore the truth values of every arithmetic constraint occurring in the guards and in the formula eventually stabilize in a computable way, which in the long run reduces the model checking of any $\NGL$-formula in an  \NGM to a model checking of an ATL-formula in a CGM. We therefore obtain the following: 

\begin{proposition}\label{prop1}
Model checking the logic $\NGL$ is decidable in any finite \NGM with non-negative payoffs. 
\end{proposition}}
 
\end{document}